\documentclass[12pt]{amsart}

\usepackage{amssymb}

\newtheorem{theorem}{Theorem}[section]
\newtheorem{proposition}[theorem]{Proposition}
\newtheorem{lemma}[theorem]{Lemma}
\newtheorem{corollary}[theorem]{Corollary}
\newtheorem{remark}[theorem]{Remark}

\numberwithin{equation}{section}

\begin{document}
\baselineskip=16pt

\title[Unitary representations of cocompact lattices of
$\text{SL}(2,{\mathbb C})$]{Solutions of Strominger 
system from unitary representations of cocompact 
lattices of $\text{SL}(2,{\mathbb C})$}

\author[I. Biswas]{Indranil Biswas}

\address{School of Mathematics, Tata Institute of Fundamental
Research, Homi Bhabha Road, Mumbai 400005, India}

\email{indranil@math.tifr.res.in}

\author[A. Mukherjee]{Avijit Mukherjee}

\address{Department of Physics,
Jadavpur University, Raja S. C. Mullick Road, Jadavpur,
Kolkata 700032, India}

\email{avijit00@gmail.com}



\date{}

\begin{abstract}
Given an irreducible unitary representation of a cocompact lattice 
of $\text{SL}(2,{\mathbb C})$, we explicitly write down a solution
of the Strominger system of equations. These solutions satisfy the
equation of motion, and the underlying holomorphic vector 
bundles are stable.
\end{abstract}

\maketitle

\section{Introduction}

Evoking physical requirements from anomaly cancellations, realistic 
fermionic spectrum and the appropriate amount ($N = 1$) of Space-time 
supersymmetry, Candelas {\it et. al} had 
originally proposed a model for compactification of the superstring, by 
analyzing the vacuum configurations of these 10-dimensional theories 
\cite{CHSW}. Anomaly cancellation requirements (which constrain the 
gauge groups of these models to be $O(32)$ or $E_8 \times E_8$), along 
with the requirement of a zero cosmological constant, then lead 
them
to propose/construct the 10-dimensional vacuum solutions of 
these theories to be of the metric product type $ X_4 \times {\mathcal 
M}$, where $X_4$ is the maximally symmetric $4d$ space-time (which 
should admit unbroken $N = 1$ supersymmetry), and $\mathcal M$ is a 
complex 3-dimensional Calabi-Yau manifold. Subsequently, these 
conclusions were further generalized to include other gauge groups (like
${\rm SU}(4)$ or ${\rm SU}(5)$), as would arise when considering compactifications
for the strongly coupled heterotic string theory. The correspondence 
between the algebro-geometric notion of stable vector bundles 
and the existence of Hermitian-Yang-Mills connections was one of the primary 
mathematical input underlying these derivations \cite{Wi1}. In all these 
examples, the supersymmetric vacuum (manifold) was assumed to be one whose
geometry had no torsion. Hence 
the existence of a solution on such a given manifold was mostly a topological 
question and the issue of existence of appropriate solutions (obeying all 
the physical requirements) often boiled down to a set of conditions on 
the Chern classes of the vacuum manifold $\mathcal M$ and the 
Yang-Mills Gauge connections.

In 1986, Strominger investigated the necessary and sufficient 
conditions for space-time supersymmetric solutions of the heterotic string.
While considering more general space-times as solutions to the heterotic superstring 
solutions, Strominger, \cite{St}, was lead to considering vacuum 
configurations {\it with} torsion. He relaxed the requirement 
of the 10-dimensional vacuum metric by considering that, for more general vacuum
configurations (which can sustain non-zero fluxes as well 
as space-time supersymmetry), the 
10-dimensional space-time be a {\it warped} product of $X_4$ and 
the 6-dimensional internal space $\mathcal M$. Analyzing the 
constraints imposed by the requirements of $N=1$ space-time ({\it 
i.e.,} 4 dimensional) supersymmetry (and other usual 
consistency requirements like anomaly cancellation), Strominger then
established that the 6-dimensional internal manifold $\mathcal 
M$ should be a compact, connected, complex manifold (hereafter denoted 
as $M$), such that its canonical line bundle $K_M$ 
is holomorphically trivial. Let $\omega \ = \
{\frac{\sqrt{-1}}{2}} g_{ij} dz^i \wedge dz^j$ be a $(1,1)$ 
Hermitian form on $M$, and let 
$\nabla^M$ be a connection on $TM$ compatible with $\omega$. We denote its curvature by $R$. Further, let $E$ be a holomorphic vector bundle on $M$ equipped with the (gauge) connection $A$, and corresponding curvature $F_A$. It turns out that the anomaly cancellation condition then demands that the Hermitian $(1,1)$ form $\omega$ obeys an equation of the form:
$$ \sqrt{-1} \, {\partial \overline\partial}  \,  \omega    \   =  \   
\frac{\alpha^\prime}{4} \  \left(  {{\rm trace}(R \wedge R)}  - 
{{\rm trace}({F_A} \wedge {F_A})} \right)\, .$$
The consistency conditions from requirements of the space-time 
supersymmetry translates into the equation:
$$    d^*  \omega   \   =  \   \sqrt{-1}    \left({\overline\partial}  -  {\partial}\right)   \ln \Vert \Omega\Vert_\omega
$$  
for the Hermitian form $\omega$ and the holomorphic 3-form $\Omega$. The previous equation may also be equivalently re-written as  \cite{LY2}: 
$$   d  \left(\Vert  \Omega \Vert_\omega  \cdot \omega^2\right)    \  =  \  0$$
  
The above equations, along with the system (constraining the Yang-Mills Gauge theory content):
$$F_A^{2,0}    \   =  \   F_A^{0,2}  \    =  \   0,  \quad  \quad  F \wedge \omega^2  \  =  \  0
$$
gives a complete and general solution of a superstring theory with torsion and
with a flux that allows a non-trivial dilation field (cosmological
constant). Henceforth, the above system of equations (which are derived solely from
the explicit requirements stemming from Superstring theory)
would be referred to as the {\it Strominger system} of 
equations. Thus, by considering vacuum geometries with 
torsion, Strominger was  able to relax the requirement of $M$ to 
be K{\"a}hler and consider more general complex 3-manifolds. But 
the price to be paid was that the familiar 
tools and methods from K\"ahler geometry could now no longer 
be applied to these more general cases. Moreover, a purely topological characterization and classification  of these heterotic
superstring vacua solutions ($i.e.$, the Chern classes of the bundles $E$ and
the vacuum manifold $M$), would no longer suffice.

The above results provide us with the necessary and sufficient 
conditions for any heterotic superstring theory solution 
(admitting space-time supersymmetry for its vacuum 
configuration) to exist, but in practice, it is quite a 
difficult matter to exhibit or actually explicitly construct a 
solution which exists (and satisfies the Strominger equation). 
Apart from its interest and usefulness in the context of string 
theory, it is also of interest from a mathematical point of 
view to find solutions ($i.e.$, construct the bundles $E$ with 
the appropriate connection $A$ for a given manifold $M$ with 
properties as defined above) of the Strominger system. In 
recent years, there has been a flurry of activities surrounding 
this problem of providing explicit constructive methods for 
solutions of these Strominger systems (cf. \cite{AF1}, \cite{AF2}, \cite{Iv} and
references therein). The present paper explores a new and altogether different constructive scheme, based on an approach that does not require the perturbative/deformation prescription.

Further attempts at exploring more general vacuum configurations for the heterotic string with non-zero fluxes have lead to some additional corrections to the original analysis of Strominger. These come from considering (${\rm SU}(3)$)
instanton corrections at higher loops, and lead to the additional consistency conditions (for the solutions of the Strominger system) and these are:
$$ R^{2,0}  \  =  \  R^{0,2}  \  =  0,  \quad  \quad  \quad  R 
\wedge  \omega^2  \  =  \  0\,  .
$$
These are referred to as equations of motion. Here
we shall consider those solutions of the
Strominger system which also additionally satisfy the above
conditions. 

In recent years, there has been a lot of activity,  in trying to 
construct actual/explicit examples which are solutions to the  
above extended {\it Strominger system}. In
\cite{FTY}, Fu, Tseng and Yau have studied 
the existence of smooth solutions to the Strominger system.
They proposed a perturbation method where 
deformation 
theory results were used to construct solutions for some $U(4)$ 
and $U(5)$ principal bundles. Subsequent generalizations of this 
method lead to the construction of new examples (of solutions to 
the Strominger system) on a class of non-K{\"a}hler 
three-dimensional manifolds like $T^2$-bundles over a $K3$ 
surface, or $T^2$-bundles over Eguchi-Hanson spaces. 
Nevertheless finding new/more examples  of such solutions has 
proved to be rather tricky,
and it seems that there is no general ansatz/scheme for 
constructing an example; instead one has to invent specific 
prescriptions and construction
procedure for every new example.

In the present work, we produce solutions of the 
Strominger system from irreducible unitary representations of
any cocompact lattice in $\text{SL}(2,{\mathbb C})$.
Let $\Gamma$ be a cocompact lattice in $\text{SL}(2,{\mathbb 
C})$ (meaning $\text{SL}(2,{\mathbb C})/\Gamma$ is compact), and 
let $\rho\, :\, \Gamma\, \longrightarrow\, 
\text{U}(n)$ be an irreducible homomorphism, meaning no
nonzero proper linear subspace of ${\mathbb C}^n$ is left
invariant by the action of the image $\rho(\Gamma)$. The compact 
complex manifold $M\, :=\,\text{SL}(2,{\mathbb C})/\Gamma$ has 
trivial canonical line
bundle, and $M$ is equipped with a natural 
Hermitian structure. The Chern connection on $TM$ for this
Hermitian structure has the following properties:
\begin{enumerate}
\item the torsion of the connection is totally skew--symmetric, meaning
it is a section of $\bigwedge^3 TM$, and

\item the holonomy of the connection lies in $\text{SU}(3)$
\end{enumerate}
(see Corollary \ref{cor-n1}). The homomorphism $\rho$ produces a
holomorphic vector bundle over $M$ with a flat unitary 
connection. This vector bundle is stable; see Proposition
\ref{prop3}. We prove that all these together produce a
solution of the Strominger system satisfying the equation of
motion; the details are in Theorem \ref{thm1}.

\section{Strominger system of equations}

We write down the Strominger system of equations in one place for 
the convenience of later reference in Section \ref{se4}.

Let $M$ be a compact connected complex manifold of dimension
three such that the canonical line bundle $K_M\, :=\, 
\bigwedge\nolimits^3 \Omega^1_M$ is holomorphically trivial. Let
$$
\Omega\, \in\, H^0(M,\, K_M)
$$
be a nowhere vanishing holomorphic section. Let $\omega$ be
a Hermitian $(1\, ,1)$--form on $M$. Take a connection 
$\nabla^T$
on $TM$ compatible with $\omega$; its curvature
will be denoted by $R$. Let $E$ be a holomorphic
vector bundle on $M$ equipped with a connection $A$. Let $F_A$
be the curvature of $A$. Let $d^*$ be the adjoint of $d$ with respect to 
$\omega$; it sends smooth $k$ forms on $M$ to $k-1$ forms.

The sextuple $(M\, ,\Omega\, ,\omega\, , \nabla^T\, , E\, ,A)$ is said 
to solve
the \textit{Strominger system} if the following equations hold:
\begin{equation}\label{st1}
F^{2,0}_A\,=\, F^{0,2}_A\,=\, 0,\, F\wedge\omega^2\,=\, 0
\end{equation}
\begin{equation}\label{st2}
d^*\omega \,=\, \sqrt{-1}(\overline{\partial}- \partial)\Vert \Omega
\Vert_\omega
\end{equation}
\begin{equation}\label{st3}
d(\Vert \Omega\Vert_\omega\cdot \omega^2)\,=\, 0
\end{equation}
\begin{equation}\label{st4}
\sqrt{-1}\partial\overline{\partial}\omega\,=\, \alpha'(\text{trace}(
R\wedge R) - \text{trace}(F_A\wedge F_A)),~ 
\, \text{where}~\, \alpha'\,\in\,
{\mathbb C}\, .
\end{equation}
A Strominger system $(M\, ,\Omega\, ,\omega\, , E\, ,A)$ as above
is said to solve the \textit{equation of motion} if
\begin{equation}\label{st5}
R^{2,0}\,= 0\, = \, R^{0,2}\, ~ \quad {\rm and} \quad ~
\,R\wedge \omega^2\,=\, 0\, .
\end{equation}

\section{Invariant forms on $\text{SL}(2,{\mathbb C})$}\label{sec2}

Consider the complex Lie group $\text{SL}(2,{\mathbb C})$.
Let $h_0$ be the Hermitian structure on the Lie algebra
$sl(2,{\mathbb C})$ of $\text{SL}(2,{\mathbb C})$ defined by
\begin{equation}\label{h0}
h_0(A, B)\, =\, \text{trace}(AB^*)\, ,
\end{equation}
where $B^*\,=\, \overline{B}^t$. Note that the adjoint action of $\text{SU}(2)$
on $sl(2,{\mathbb C})$ preserves $h_0$.

Using the right--translation invariant vector fields on 
$\text{SL}(2,{\mathbb C})$, we identify the holomorphic tangent
bundle $T\text{SL}(2,{\mathbb C})$ with the trivial vector
bundle $$\text{SL}(2,{\mathbb C})\times sl(2,{\mathbb C})\,\, 
\longrightarrow\, \text{SL}(2,{\mathbb C})$$
with fiber $sl(2,{\mathbb C})$. Let $h$ be the
unique right--translation invariant Hermitian structure on 
$\text{SL}(2,{\mathbb C})$ such that
$$
h\vert_{T_e\text{SL}(2,{\mathbb C})}\, =\, h_0\, ,$$
where
$e\, \in\, \text{SL}(2,{\mathbb C})$ is the identity element.
Let
\begin{equation}\label{oh}
\omega_h\, \in\, C^\infty(\text{SL}(2,{\mathbb C}), \,
\Omega^{1,1}_{\text{SL}(2,{\mathbb C})})
\end{equation}
be the K\"ahler form associated to the Hermitian structure $h$
on $\text{SL}(2,{\mathbb C})$. We note that $d\omega_h\,\not=\, 0$.

\begin{proposition}\label{prop1}
Let $\xi\, \in\, C^\infty({\rm SL}(2,{\mathbb C}),\,
\Omega^{1,0}_{{\rm SL}(2,{\mathbb C})}\oplus
\Omega^{0,1}_{{\rm SL}(2,{\mathbb C})})$ be a complex $1$--form
on ${\rm SL}(2,{\mathbb C})$ such that
\begin{itemize}
\item the right--translation action of ${\rm SL}(2,{\mathbb C})$
on itself preserves $\xi$, and

\item the left--translation action of ${\rm SU}(2)$
on ${\rm SL}(2,{\mathbb C})$ preserves $\xi$.
\end{itemize}
Then
$$
\xi\,=\, 0\, .
$$
\end{proposition}

\begin{proof}
Since the holomorphic tangent space of $\text{SL}(2,{\mathbb C})$
at $e\, \in\, \text{SL}(2,{\mathbb C})$ is identified with 
$sl(2,{\mathbb C})$, the evaluation of
$\xi$ at $e$ is an element of
$sl(2,{\mathbb C})^*\bigotimes_{\mathbb R}{\mathbb C}\,=\,
(sl(2,{\mathbb C})\bigotimes_{\mathbb R}{\mathbb C})^*$;
here we identify $(T^{0,1}_e\text{SL}(2,{\mathbb C}))^*$ with
$(T^{1,0}_e\text{SL}(2,{\mathbb C}))^*$ by sending any $u$ to
its conjugate $\overline{u}$. Let
$$
\xi_0\, :=\, \xi(e) \, \in\, sl(2,{\mathbb C})^*\otimes_{\mathbb 
R}{\mathbb C}
$$
be the evaluation of $\xi$ at $e$. The adjoint action of 
$\text{SL}(2,{\mathbb C})$ on $sl(2,{\mathbb C})$ produces an
action of $\text{SL}(2,{\mathbb C})$ on $sl(2,{\mathbb 
C})^*\bigotimes_{\mathbb R}{\mathbb C}$. In particular, we get
an action of ${\rm SU}(2)$ on $sl(2,{\mathbb
C})^*\bigotimes_{\mathbb R}{\mathbb C}$. The two given conditions
on $\xi$ imply that this action of ${\rm SU}(2)$ on
$sl(2,{\mathbb C})^*\bigotimes_{\mathbb R}{\mathbb C}$ fixes
the element $\xi_0$.

Consider the nondegenerate symmetric
bilinear pairing on $sl(2,{\mathbb C})$ defined by
\begin{equation}\label{tr}
(A\, ,B) \, \longmapsto\, \text{trace}(AB)\, .
\end{equation}
It produces an isomorphism of $sl(2,{\mathbb C})$ with
$sl(2,{\mathbb C})^*$ that is equivariant for the actions of
$\text{SL}(2,{\mathbb C})$ on $sl(2,{\mathbb C})$ and
$sl(2,{\mathbb C})^*$. Using this identification between
$sl(2,{\mathbb C})^*$ and $sl(2,{\mathbb C})$, the 
above element $\xi_0$ gives an element
$$
\widetilde{\xi}_0\, \in\, \, \in\, sl(2,{\mathbb C})\otimes_{\mathbb
R}{\mathbb C}\, .
$$
We note that $\widetilde{\xi}_0$ is fixed by the adjoint action of
${\rm SU}(2)$, because
\begin{itemize}
\item $\xi_0$ is fixed by the action of ${\rm SU}(2)$ on
$sl(2,{\mathbb C})^*\otimes_{\mathbb R}{\mathbb C}$, and

\item the isomorphism between $sl(2,{\mathbb C})$ and
$sl(2,{\mathbb C})^*$ is $\text{SL}(2,{\mathbb C})$--equivariant.
\end{itemize}

But no nonzero element of $sl(2,{\mathbb C})$ is fixed by the adjoint
action of ${\rm SU}(2)$ on $sl(2,{\mathbb C})$. This implies that there
is no nonzero element of $sl(2,{\mathbb C})\bigotimes_{\mathbb R}
{\mathbb C}$ that is fixed by the action of ${\rm SU}(2)$, because
$(sl(2,{\mathbb C})\bigotimes_{\mathbb R}{\mathbb C})^{{\rm SU}(2)}\,=\,
sl(2,{\mathbb C})^{{\rm SU}(2)}\bigotimes_{\mathbb R}{\mathbb C}$.
(For an ${\rm SU}(2)$--module $W$, by $W^{{\rm SU}(2)}$ we 
denote the space of invariants for the action of
${\rm SU}(2)$ on $W$.) Hence
we conclude that $\widetilde{\xi}_0\,=\,0$. So, $\xi_0\,=\, 0$. This 
implies that $\xi\,=\,0$ because it is fixed by the right--translation 
action of ${\rm SL}(2,{\mathbb C})$ on itself.
\end{proof}

\begin{proposition}\label{prop2}
Let $\zeta$ be a $C^\infty$ complex $4$--form on ${\rm SL}(2,{\mathbb 
C})$ such that
\begin{itemize}
\item the right--translation action of ${\rm SL}(2,{\mathbb C})$
on itself preserves $\zeta$, and

\item the left--translation action of ${\rm SU}(2)$
on ${\rm SL}(2,{\mathbb C})$ preserves $\zeta$.
\end{itemize}
Then there is constant $c\, \in\, \mathbb C$ such that
$$
\zeta\,=\, c\cdot \omega_h\wedge\omega_h\, ,
$$
where $\omega_h$ is constructed in \eqref{oh}.
\end{proposition}

\begin{proof}
As in the proof of Proposition \ref{prop1}, the evaluation of
$\zeta$ at $e$ is an element
$$
\zeta_0\, \in\, \bigwedge\nolimits^4 (sl(2,{\mathbb
C})^*\otimes_{\mathbb R}{\mathbb C})\, .
$$

The adjoint action of $\text{SL}(2,{\mathbb C})$ on $sl(2,{\mathbb C})$
produces an action of $\text{SL}(2,{\mathbb C})$ on the complex line
$\bigwedge^6 (sl(2,{\mathbb C})\bigotimes_{\mathbb R}{\mathbb C})$.
Since $\text{SL}(2,{\mathbb C})$ does not have any nontrivial character,
this action of $\text{SL}(2,{\mathbb C})$ on $\bigwedge^6 (sl(2,{\mathbb 
C})\bigotimes_{\mathbb R}{\mathbb C})$ is trivial.
The adjoint action of $\text{SL}(2,{\mathbb 
C})$ on the Lie algebra
$sl(2,{\mathbb C})$ produces actions of $\text{SL}(2,{\mathbb 
C})$ on $\bigwedge^4 (sl(2,{\mathbb C})^*\bigotimes_{\mathbb 
R}{\mathbb C})$ and $\bigwedge^2(sl(2,{\mathbb 
C})\bigotimes_{\mathbb R}{\mathbb C})$. Fixing a nonzero
element of the line $\bigwedge^6 (sl(2,{\mathbb C})\bigotimes_{\mathbb 
R}{\mathbb 
C})$, we get an $\text{SL}(2,{\mathbb C})$--equivariant isomorphism of 
$\bigwedge^4 (sl(2,{\mathbb C})^*\bigotimes_{\mathbb R}{\mathbb C})$ 
with $\bigwedge^2 (sl(2,{\mathbb C})\bigotimes_{\mathbb 
R}{\mathbb C})$.
Using this isomorphism, the above element $\zeta_0$ gives an element
\begin{equation}\label{e4}
\widehat{\zeta}_0\,\in\, \bigwedge\nolimits^2 (sl(2,{\mathbb 
C})\otimes_{\mathbb R}{\mathbb C})\, .
\end{equation}

The two given conditions on $\zeta$ imply that the 
element $\widehat{\zeta}_0$ in \eqref{e4} is fixed by the action
of $\text{SU}(2)$ on $\bigwedge^2(sl(2,{\mathbb C})\bigotimes_{\mathbb 
R}{\mathbb C})$ (recall that ${\rm SL}(2,{\mathbb C})$
acts on $\bigwedge^2(sl(2,{\mathbb C})\bigotimes_{\mathbb
R}{\mathbb C})$).

Note that
$$
\bigwedge\nolimits^2 (sl(2,{\mathbb C})\otimes_{\mathbb R}{\mathbb 
C})\,=\,
(\bigwedge\nolimits^2 sl(2,{\mathbb C}))^{\oplus 2} \oplus 
(sl(2,{\mathbb C})\otimes sl(2,{\mathbb C}))\, ;
$$
this decomposition is preserved by the action of $\text{SL}(2,{\mathbb 
C})$. There is no nonzero element of $\bigwedge^2 sl(2,{\mathbb C})$
preserved by the action of $\text{SU}(2)$. The subspace of
$sl(2,{\mathbb C})\otimes sl(2,{\mathbb C})$ defined by all
elements fixed pointwise by the action of $\text{SU}(2)$ is
one-dimensional, and it is generated by the element of
$\text{Sym}^2(sl(2,{\mathbb C}))\,\subset\,
sl(2,{\mathbb C})^{\otimes 2}$ given by the 
nondegenerate pairing
in \eqref{tr}. This immediately implies that the space of smooth complex
$4$--forms on $\text{SL}(2,{\mathbb C})$ satisfying the two conditions
in the proposition is one dimensional.

Since the inner product $h_0$ on $sl(2,{\mathbb C})$ in \eqref{h0} is
$\text{SU}(2)$--invariant, it follows immediately that the 
Hermitian
structure $h$ on $\text{SL}(2,{\mathbb C})$ is preserved by the
left--translation action of
$\text{SU}(2)$ on $\text{SL}(2,{\mathbb C})$. Hence the
K\"ahler form $\omega_h$ on $\text{SL}(2,{\mathbb C})$ is preserved by 
the left--translation action
of $\text{SU}(2)$ on $\text{SL}(2,{\mathbb C})$. Recall that
$\omega_h$ is also preserved by the right--translation action
of $\text{SL}(2,{\mathbb C})$ on itself. Therefore, $\omega_h\bigwedge \omega_h$
is a nonzero complex $4$--form satisfying the two conditions in the
proposition. Since the space of smooth 
complex $4$--forms on $\text{SL}(2,{\mathbb C})$ satisfying the two 
conditions in the proposition is one dimensional, we now conclude
that $\zeta$ is a 
constant scalar multiple of $\omega_h\bigwedge \omega_h$.
\end{proof}

\begin{lemma}\label{lem1}
The differential form $\omega_h$ in \eqref{oh} satisfies the identity
$$
d(\omega_h\wedge \omega_h)\, =\, 0\, .
$$
\end{lemma}

\begin{proof}
Using the identification between $T_e\text{SL}(2,{\mathbb C})$
and $sl(2,{\mathbb C})$, the evaluation of
the $5$--form $d(\omega^2_h)$ at $e$ is an element of
$\bigwedge^5 (sl(2,{\mathbb C})\bigotimes_{\mathbb R}{\mathbb C})^*$;
as in the proof of Proposition \ref{prop1}, we identify 
$(T^{0,1}_e\text{SL}(2,{\mathbb C}))^*$ with
$(T^{1,0}_e\text{SL}(2,{\mathbb C}))^*$ by sending any $u$ to
$\overline{u}$.

As in the proof of Proposition \ref{prop2}, fixing a nonzero
element of $\bigwedge^6 (sl(2,{\mathbb C})\bigotimes_{\mathbb 
R}{\mathbb 
C})$, we get an $\text{SL}(2,{\mathbb C})$--equivariant isomorphism of 
$\bigwedge^5 (sl(2,{\mathbb C})\bigotimes_{\mathbb R}{\mathbb C})^*$ with
$sl(2,{\mathbb C})\bigotimes_{\mathbb R}{\mathbb C}$. Using this
isomorphism, we have
\begin{equation}\label{el}
(d(\omega_h\wedge \omega_h))(e) \,\in\, sl(2,{\mathbb 
C})\otimes_{\mathbb R}
{\mathbb C}\, .
\end{equation}

As noted in the proof of Proposition \ref{prop2}, the 
K\"ahler form $\omega_h$ is preserved by the left--translation action 
of $\text{SU}(2)$ on $\text{SL}(2,{\mathbb C})$. Consequently,
the $5$--form $d(\omega^2_h)$ is preserved by the left--translation 
action of $\text{SU}(2)$ on $\text{SL}(2,{\mathbb C})$. This implies 
that the element $(d(\omega^2_h))(e)$ in \eqref{el}
is fixed by the adjoint action of $\text{SU}(2)$ on 
$sl(2,{\mathbb C})\bigotimes_{\mathbb R}{\mathbb C}$. From
this it follows that $(d(\omega^2_h))(e)\,=\, 0$, because
$(sl(2,{\mathbb C})\bigotimes_{\mathbb R}{\mathbb 
C})^{\text{SU}(2)}\,=\, sl(2,{\mathbb C})^{\text{SU}(2)} 
\bigotimes_{\mathbb R}{\mathbb C}\,=\, 0$.
Since $d\omega^2_h$ is invariant under the
right--translation action of $\text{SL}(2,{\mathbb C})$ on itself,
and $(d(\omega^2_h))(e)\,=\, 0$, we conclude that $d(\omega^2_h)\,=\,0$.
\end{proof}

As before, $T\text{SL}(2,{\mathbb C})$ is the holomorphic tangent bundle
of $T\text{SL}(2,{\mathbb C})$. Let $\nabla^h$ denote the Chern connection
on $T\text{SL}(2,{\mathbb C})$ corresponding to the Hermitian structure $h$
on $\text{SL}(2,{\mathbb C})$. The torsion of the connection $\nabla^h$
of $T\text{SL}(2,{\mathbb C})$ will be denoted by ${\mathcal T}(\nabla^h)$;
it is a $C^\infty$ section of $\Omega^{2,0}_{\text{SL}(2,{\mathbb C})}\otimes
(T\text{SL}(2,{\mathbb C}))$.

Consider the Hermitian structure $h$ on $T\text{SL}(2,{\mathbb C})$. It produces
a $C^\infty$ isomorphism
$$
h'\, :\, \Omega^{1,0}_{\text{SL}(2,{\mathbb C})}\, \longrightarrow\,
T\text{SL}(2,{\mathbb C})
$$
defined by $h(h'(w)\, ,v)\,=\, w(v)$ for $w\,\in\,
(\Omega^{1,0}_{\text{SL}(2,{\mathbb C})})_x$, $v\, \in\,T_x\text{SL}(2,{\mathbb C})$
and $x\, \in\, \text{SL}(2,{\mathbb C})$. We note that $h'$ is a conjugate
linear isomorphism. Using the isomorphism $h'$, the torsion
${\mathcal T}(\nabla^h)$ is a $C^\infty$ section of
$(\bigwedge^2(T\text{SL}(2,{\mathbb C})))\bigotimes (T\text{SL}(2,{\mathbb C}))$.

\begin{proposition}\label{prop-n1}
The torsion ${\mathcal T}(\nabla^h)\,\in\, C^\infty({\rm SL}(2,{\mathbb C}),
\, (\bigwedge\nolimits^2(T{\rm SL}(2,{\mathbb C})))\bigotimes
(T{\rm SL}(2,{\mathbb C})))$ lies in the subspace
$$
C^\infty({\rm SL}(2,{\mathbb C}),\,
\bigwedge\nolimits^3 (T{\rm SL}(2,{\mathbb C}))) \,\subset\,
C^\infty({\rm SL}(2,{\mathbb C}),\, (\bigwedge\nolimits^2
(T{\rm SL}(2,{\mathbb C})))\otimes (T{\rm SL}(2,{\mathbb C})))\, .
$$
In other words, the torsion is totally skew--symmetric.

The holonomy of the connection $\nabla^h$ lies in ${\rm SU}(3)$.
\end{proposition}

\begin{proof}
Consider the element
\begin{equation}\label{an1}
{\mathcal T}(\nabla^h)(e)\,\in\, (\bigwedge\nolimits^2 sl(2,{\mathbb C}))
\otimes sl(2,{\mathbb C})\, ,
\end{equation}
where $e\, \in\, {\rm SL}(2,{\mathbb C})$ is the identity element.
It is invariant under the adjoint action of $\text{SU}(2)$ because the
Hermitian structure $h$ is preserved by the lest translation action of
$\text{SU}(2)$ on ${\rm SL}(2,{\mathbb C})$.

Let $V_0$ be the standard two dimensional representation of $\text{SU}(2)$.
The $\text{SU}(2)$--module $sl(2,{\mathbb C})$ is isomorphic to the symmetric
product $\text{Sym}^2(V_0)$.

Therefore, the $\text{SU}(2)$--module in \eqref{an1} is isomorphic to
$(\bigwedge\nolimits^2 \text{Sym}^2(V_0))\otimes \text{Sym}^2(V_0)$. But
$$
\bigwedge\nolimits^2 \text{Sym}^2(V_0)\,=\, \text{Sym}^2(V_0)
$$
(see \cite[p. 160, Ex. 11.35]{FH}), and 
$$
\text{Sym}^2(V_0)\otimes \text{Sym}^2(V_0)\,=\, \text{Sym}^4(V_0)\oplus
\text{Sym}^2(V_0) \oplus \text{Sym}^0(V_0)
$$
(see \cite[p. 151, Ex. 11.11]{FH}). Consequently,
$$
((\bigwedge\nolimits^2 \text{Sym}^2(V_0))\otimes \text{Sym}^2(V_0))^{\text{SU}(2)}
\,=\, \text{Sym}^0(V_0)\,=\, \bigwedge\nolimits^3 \text{Sym}^2(V_0)\, .
$$
Consequently, ${\mathcal T}(\nabla^h)$ is a section of
$\bigwedge\nolimits^3 (T{\rm SL}(2,{\mathbb C}))$.
This proves the first part of the proposition.

To prove the second part of the proposition, consider the Hermitian structure on
the trivial holomorphic line bundle $\bigwedge\nolimits^3 (T{\rm SL}(2,{\mathbb C}))$
induced by $h$. It is a constant Hermitian structure on the
trivial holomorphic line bundle. Hence the holonomy of the connection on
$\bigwedge\nolimits^3 (T{\rm SL}(2,{\mathbb C}))$ induced by $\nabla^h$ is
trivial. Consequently, the holonomy of the connection $\nabla^h$ lies in
the subgroup ${\rm SU}(3)\, \subset\, \text{U}(3)$.
\end{proof}

\section{A class of solutions of the Strominger 
system}\label{se4}

Let
\begin{equation}\label{Ga}
\Gamma\, \subset\, \text{SL}(2,{\mathbb C})
\end{equation}
be a cocompact lattice, meaning $\Gamma$ is a closed discrete subgroup
of $\text{SL}(2,{\mathbb C})$ such that the quotient
\begin{equation}\label{M}
M\, :=\, \text{SL}(2,{\mathbb C})/\Gamma
\end{equation}
is compact. We note that $M$ is {\it not} a K\"ahler manifold.

Since the Hermitian structure $h$ on $\text{SL}(2,{\mathbb C})$
constructed in Section \ref{sec2} is invariant under the
right--translation action of $\text{SL}(2,{\mathbb C})$ on itself, we
conclude that $h$ defines a Hermitian structure on $M$. Let 
$\widehat{h}$ denote the Hermitian structure on
$M$ given by $h$. Note that the pullback of 
$\widehat{h}$ by the quotient map $\text{SL}(2,{\mathbb 
C})\,\longrightarrow\, M$ coincides with $h$. Let
\begin{equation}\label{om}
\omega\, \in\, C^\infty(M, \, \Omega^{1,1}_M)
\end{equation}
be the K\"ahler form on $M$ associated to $\widehat{h}$. Let
\begin{equation}\label{cc}
\nabla^\omega
\end{equation}
be the Chern connection on $TM$ associated to $\omega$.

\begin{corollary}\label{cor1}
The differential form $\omega$ in \eqref{om} satisfies the identity
$$
d(\omega^2)\, =\, 0\, .
$$
\end{corollary}

\begin{proof}
Since the pullback of $\omega$ to $\text{SL}(2,{\mathbb C})$, by
the quotient map $\text{SL}(2,{\mathbb C})\, \longrightarrow\, M$,
coincides with $\omega_h$, from Lemma \ref{lem1} it follows that
$d(\omega^2)\, =\, 0$.
\end{proof}

For any torsionfree coherent analytic sheaf $F$ on $M$, let
$\det (F)$ be the determinant line bundle on $M$; see
\cite[Ch.~V, \S~6]{Ko} for the construction of the determinant bundle.
Define the \textit{degree} of $F$ to be
\begin{equation}\label{de}
{\rm degree}(F)\, :=\, \int_M \alpha(F)\wedge \omega\wedge \omega
\, \in\, {\mathbb R}\, ,
\end{equation}
where $\alpha(F)$ is any $2$--form on $M$ representing the
first Chern class $c_1(\det (F))\, \in\, H^2(M,\, {\mathbb R})$.

\begin{lemma}\label{lem2}
The degree is well defined.
\end{lemma}

\begin{proof}
Let $\alpha$ and $\beta$ be two $2$--forms on $M$ representing
$c_1(\det (F))$. So, $\alpha-\beta\,=\, d\delta$, where $\delta$ is
a smooth $1$--form on $M$. Now,
$$
\int_M \alpha\wedge \omega^2 - \int_M \beta\wedge \omega^2
\,=\, \int_M (\alpha -\beta) \wedge \omega^2
\,=\, \int_M (d\delta)\wedge \omega^2\,=\,
\int_M \delta\wedge d(\omega^2)\, =\, 0
$$
be Corollary \ref{cor1}. So, $\int_M \alpha\wedge \omega^2 \,=\, \int_M 
\beta\wedge \omega^2$. Hence the degree is independent of the choice
of the differential form representing the first Chern class.
\end{proof}

Since the connection $\nabla^\omega$ is the descent of the connection $\nabla^h$
considered in Proposition \ref{prop-n1},
the following corollary is an immediate consequence of Proposition \ref{prop-n1}.

\begin{corollary}\label{cor-n1}
The torsion of the connection $\nabla^\omega$ is
a $C^\infty$ section of $\bigwedge\nolimits^3 TM$; in other words,
the torsion is totally skew--symmetric.

The holonomy of the connection $\nabla^\omega$ lies in ${\rm SU}(3)$.
\end{corollary}

We note that the torsion of the connection $\nabla^\omega$ is nonzero because
$M$ is not K\"ahler.

We choose $\Gamma$ such that there are irreducible unitary 
representations of $\Gamma$.

\begin{remark}\label{rem-ex}
{\rm There are many examples of such $\Gamma$; see \cite[p. 3393,
Theorem 2.1]{La}. Note that any free nonabelian group has irreducible
unitary representations in ${\rm U}(n)$ for all $n\, \geq\, 2$.
To see this, take any two elements $g_1$ and $g_2$ of ${\rm SU}(n)$ such 
that $g_1g_2g^{-1}_1g^{-1}_2$ is a generator of the center of
${\rm SU}(n)$. The subgroup of ${\rm U}(n)$ generated by
$g_1$ and $g_2$ is irreducible.}
\end{remark}

Let
\begin{equation}\label{rho}
\rho\, :\, \Gamma\, \longrightarrow\, \text{U}(n)
\end{equation}
be an irreducible representation; this means that the only
linear subspaces of ${\mathbb C}^n$ left invariant by the
action of $\rho(\Gamma)$ are $0$ and ${\mathbb C}^n$. Let
\begin{equation}\label{en}
(E\, ,\nabla)\,\longrightarrow\, M
\end{equation}
be the
unitary flat vector bundle over $M$ given by $\rho$. We briefly
recall the constructions of the vector bundle $E$ and the connection 
$\nabla$ on it. Consider the trivial vector bundle
$\text{SL}(2,{\mathbb C})\times {\mathbb C}^n$ on $\text{SL}(2,{\mathbb 
C})$; it has the trivial connection. This trivial connection is unitary
with respect to the standard inner product on ${\mathbb C}^n$. The
group $\Gamma$ acts on $\text{SL}(2,{\mathbb C})$ as 
right--translations, and it acts on
${\mathbb C}^n$ as follows: the action of any $\gamma\, \in\, \Gamma$
sends any $v\, \in\, {\mathbb C}^n$ to $\rho(\gamma^{-1})(v)$. Consider
the diagonal action of $\Gamma$ on $\text{SL}(2,{\mathbb C})\times 
{\mathbb C}^n$ constructed using these two actions. Let
$(\text{SL}(2,{\mathbb C})\times {\mathbb C}^n)/\Gamma$ be the
quotient for this action. The natural map
$$
(\text{SL}(2,{\mathbb C})\times {\mathbb C}^n)/\Gamma\,\longrightarrow
\,\text{SL}(2,{\mathbb C})/\Gamma\, = \, M
$$
is a vector bundle, which we will denote by $E$. The trivial connection
on the vector bundle $\text{SL}(2,{\mathbb C})\times {\mathbb C}^n
\,\longrightarrow\,\text{SL}(2,{\mathbb C})$ descends to a flat
unitary connection on $E$; this descended connection on $E$ will be
denoted by $\nabla$.

A holomorphic vector bundle $F$ of positive rank on $M$ is called 
\textit{stable} if for every nonzero coherent analytic subsheaf $V\, 
\subset\, F$ with
$\text{rank}(V)\, <\, \text{rank}(F)$, the inequality
$$
\frac{\text{degree}(V)}{\text{rank}(V)}\, <\,
\frac{\text{degree}(F)}{\text{rank}(F)}
$$
holds, where degree is defined in \eqref{de} (and Lemma \ref{lem2}).

\begin{proposition}\label{prop3}
The holomorphic vector bundle $E$ over $M$ in \eqref{en} is stable.
\end{proposition}

\begin{proof}
Since the vector bundle $E$ admits a flat connection (recall that 
$\nabla$ is flat), we have $c_1(\det 
(E)) \,=\, c_1(E)\,=\, 0$. Hence $\text{degree}(E)\,=\, 0$.

Since the connection $\nabla$ in \eqref{en} is unitary flat and
irreducible, the proof of Proposition 8.2 in \cite[page 176]{Ko}
gives that $E$ is stable. In fact, the proof of Proposition 8.2 in 
\cite[page 176]{Ko}, which is for irreducible Einstein-Hermitian
bundles, gets simplified due to the stronger input that 
$\nabla$ is unitary flat.
\end{proof}

Let $\{A_0\, ,B_0\, ,C_0\}$ be the basis of $sl(2,{\mathbb C})$ 
defined by
$$
A_0\,=\,
\begin{pmatrix}
1 & 0\\
0 & -1
\end{pmatrix}\, , ~ B_0\,=\,
\begin{pmatrix}
0 & 1\\
0 & 0
\end{pmatrix}\, , ~ C_0\,=\,
\begin{pmatrix}
0 & 0\\
1 & 0
\end{pmatrix}\, .
$$
Then $A_0\bigwedge B_0\bigwedge C_0$ is a nonzero element of 
the line $\bigwedge^3 sl(2,{\mathbb C})$; we will call this 
element $\theta_0$. Note that the adjoint action of
$\text{SL}(2,{\mathbb C})$ on $\bigwedge^3 sl(2,{\mathbb C})$ preserves 
$\theta_0$, because the action of $\text{SL}(2,{\mathbb C})$ on 
$\bigwedge^3 sl(2,{\mathbb C})$ is trivial (the group 
$\text{SL}(2,{\mathbb C})$ does not have any nontrivial character).

The holomorphic tangent bundle $T\text{SL}(2,{\mathbb C})$ of
$\text{SL}(2,{\mathbb C})$ is 
identified with the trivial vector bundle $\text{SL}(2,{\mathbb 
C})\times sl(2,{\mathbb C})$ using right--translation invariant vector
fields. This identification produces a holomorphic isomorphism of the 
holomorphic tangent bundle $TM$, where $M$ is
constructed in \eqref{M}, with the trivial vector 
bundle $M\times sl(2,{\mathbb C})$. Using this isomorphism, the above 
element $\theta_0\, \in\, 
\bigwedge^3 sl(2,{\mathbb C})$ produces a trivialization of the 
canonical line bundle
$$
K_M\, :=\, \bigwedge\nolimits^3 \Omega^{3,0}_M\,=\,
(\bigwedge\nolimits^3 TM)^*\, .
$$
Let
\begin{equation}\label{theta}
\theta\, \in\, H^0(M,\, K_M)
\end{equation}
be the nowhere zero holomorphic section given by $\theta_0$.

\begin{theorem}\label{thm1}
Consider the sextuple $(M\, ,\theta\, ,\omega\, ,
\nabla^\omega\, , E\, ,\nabla)$
constructed in \eqref{M}, \eqref{theta}, \eqref{om}, \eqref{cc} and
\eqref{en}. It solves the Strominger system. Moreover, it
solves the equation of motion.
\end{theorem}

\begin{proof}
Since $\nabla$ is flat, the equations in \eqref{st1} are satisfied.

The differential forms on both sides of equation \eqref{st2} are 
given by 
right--translation invariant $1$--forms on $\text{SL}(2,{\mathbb C})$.
Moreover, these two $1$--forms on $\text{SL}(2,{\mathbb C})$ are
invariant under the left--translation action of $\text{SU}(2)$.
Hence both sides of equation \eqref{st2} vanish identically by
Proposition \ref{prop1}.

The two form $\Vert \Omega\Vert_\omega\cdot \omega^2$ is
given by a
right--translation invariant $1$--forms on $\text{SL}(2,{\mathbb C})$
which is also fixed by the left--translation action of $\text{SU}(2)$
on $\text{SL}(2,{\mathbb C})$. Therefore, by Proposition 
\ref{prop2}, the form $\Vert \Omega\Vert_\omega\cdot 
\omega^2$ is a constant scalar multiple of $\omega^2$. Hence
$d(\Vert \Omega\Vert_\omega\cdot \omega^2)\,=\, 0$ by Corollary
\ref{cor1}.

The two $2$--forms on two sides of equation \eqref{st4} are 
given by
right--translation invariant $2$--forms on $\text{SL}(2,{\mathbb C})$
that are fixed by the left--translation action of $\text{SU}(2)$
on $\text{SL}(2,{\mathbb C})$. Therefore, from Proposition \ref{prop2} 
we conclude that \eqref{st4} holds.

Therefore, the sextuple $(M\, ,\theta\, ,\omega\, , \nabla^\omega
\, , E\, ,\nabla)$ solves the Strominger system. We will now show
that equation \eqref{st5} also holds.

Let $R(\nabla^\omega)$ be the curvature of the connection
$\nabla^\omega$ on $TM$. Since $\nabla^\omega$ is the Chern connection 
for $\omega$, we have
$$
R(\nabla^\omega)^{2,0}\,=\,0\,=\, R(\nabla^\omega)^{0,2}\, .
$$

To prove that $R(\nabla^\omega)\bigwedge \omega^2\,=\, 0$, we first
note that $R(\nabla^\omega)\bigwedge \omega^2\,=\, 0$ if and only if
$$
\star_\omega (R(\nabla^\omega)\wedge \omega^2)\, =\, 0\, ,
$$
where $\star_\omega$ is the star operator on differential forms
on $M$ constructed using $\omega$; we note that $\star_\omega (R(\nabla^\omega)
\bigwedge \omega^2)$ is a $C^\infty$ section of $\text{End}(TM)\,=\,
TM \otimes (TM)^*$. Using the identification of 

Consider the evaluation $\star_\omega (R(\nabla^\omega)\bigwedge 
\omega^2)(e)\, \in\,\text{End}(T_e M)$ of $\star_\omega 
(R(\nabla^\omega)\bigwedge \omega^2)$ at $e\, \in\,
\text{SL}(2,{\mathbb C})$. Using the identification of
$T_e M$ with $sl(2,{\mathbb C})$, it will be considered
as an element of
$$
\text{End}(sl(2,{\mathbb C}))\,=\,
sl(2,{\mathbb C})\otimes sl(2,{\mathbb C})^*\, .
$$
The space of invariants 
$\text{End}(sl(2,{\mathbb C}))^{\text{SU}(2)}\, \subset\,
\text{End}(sl(2,{\mathbb C}))$ is one dimensional, and it is 
generated by the identity element $\text{Id}_{sl(2,{\mathbb 
C})}$. In other words, $\star_\omega 
(R(\nabla^\omega)\bigwedge \omega^2)(e)$ is a scalar multiple of 
$\text{Id}_{sl(2,{\mathbb C})}$. Let $\lambda\, \in\, \mathbb C$
be such that
\begin{equation}\label{la}
\star_\omega (R(\nabla^\omega)\bigwedge \omega^2)(e)\, =\, \lambda\cdot
\text{Id}_{sl(2,{\mathbb C})}\, .
\end{equation}

Since $\star_\omega
(R(\nabla^\omega)\bigwedge \omega^2)$ is given by a
section of $\text{End}(T\text{SL}(2,{\mathbb C}))$ which is invariant
under the right--translation action of $\text{SL}(2,{\mathbb C})$ on
itself, from \eqref{la} we conclude that
\begin{equation}\label{la2}
\star_\omega (R(\nabla^\omega)\bigwedge \omega^2)(e)\, =\, \lambda\cdot
\text{Id}_{TM}\, .
\end{equation}
{}From \eqref{la2} it follows immediately that
\begin{equation}\label{la3}
R(\nabla^\omega)\bigwedge \omega^2\,=\, \lambda\cdot \text{Id}_{TM}\otimes
\omega^3\, .
\end{equation}

Since $\star_\omega (R(\nabla^\omega)\bigwedge \omega^2)$ is given by a
section of $\text{End}(T\text{SL}(2,{\mathbb C}))$ which is invariant
under the right--translation action of $\text{SL}(2,{\mathbb C})$ on
itself, to prove that $\star_\omega (R(\nabla^\omega)\bigwedge \omega^2)
\,=\, 0$, it suffices to show that $\lambda\, =\, 0$, where 
$\lambda$ is the scalar in \eqref{la}.

To prove that $\lambda\, =\, 0$, first that $c_1(TM)\,=\, 0$, 
because $TM$ is holomorphically trivial. Hence
$$
\text{trace}(R(\nabla^\omega))\, =\, d\beta
$$
for some smooth $1$--form $\beta$ on $M$.
Therefore,
\begin{equation}\label{la4}
\int_M \text{trace}(R(\nabla^\omega))\wedge \omega^2\,=\, \int_M 
(d\beta)\wedge 
\omega^2\,=\, \int_M \beta\wedge d(\omega^2)\,=\, 0
\end{equation}
by Lemma \ref{lem1}. Now from \eqref{la3},
$$
\int_M \text{trace}(R(\nabla^\omega))\wedge \omega^2
\, =\, 3\lambda\cdot \int_M \omega^3\, .
$$
Since $\int_M \omega^3\, \not=\, 0$, from \eqref{la4} we
conclude that $\lambda\, =\, 0$. Therefore, \eqref{st5}
holds. This completes the proof.
\end{proof}

\medskip
\noindent
\textbf{Acknowledgements.}\, We thank Mahan Mj for pointing out
\cite{La}. The first--named author wishes to thank the Indian
Statistical Institute at Kolkata for providing hospitality 
while the work was carried out.

\end{document}